\documentclass[journal]{IEEEtran} 
\ifCLASSINFOpdf
  \else
  \fi
\usepackage{bm}
\usepackage{color}
\usepackage{setspace}
\usepackage{graphicx}
\usepackage{subfigure}
\usepackage{psfrag}
\usepackage{algorithmic}
\usepackage{algorithm}
\usepackage{comment}
\usepackage{url}
\usepackage{times}
\usepackage{lpic}
\usepackage{amsmath}
\usepackage{cases}
\usepackage{amsthm} 
\usepackage{cleveref}
\usepackage{epsfig}
\usepackage{longtable}
\usepackage{amsfonts}
\usepackage{amssymb}%
\usepackage{tabularx}
\usepackage{ifpdf}
\usepackage{multirow}
\usepackage{epstopdf}
\usepackage{subfigure}

\newcommand{\ud}{\,\mathrm{d}}

\newtheorem{defn}{Definition}
\newtheorem{theorem}{Theorem}

\begin{document}

\title{Energy Harvesting-Aided Spectrum Sensing and Data Transmission in Heterogeneous Cognitive Radio Sensor Network}
%
\author{\normalsize{
Deyu Zhang,
Zhigang Chen,~\IEEEmembership{ Member,~IEEE,}
Ju Ren,~\IEEEmembership{Student Member,~IEEE,}
Ning Zhang,~\IEEEmembership{Member,~IEEE,}
Mohamad Khattar Awad,~\IEEEmembership{Member,~IEEE,}
Haibo Zhou,~\IEEEmembership{ Member,~IEEE,}
Xuemin (Sherman) Shen,~\IEEEmembership{Fellow,~IEEE}

\thanks{Copyright (c) 2015 IEEE. Personal use of this material is permitted. However, permission to use this material for any other purposes must be obtained from the IEEE by sending a request to pubs-permissions@ieee.org.

D. Zhang and J. Ren are with the School of Information Science and Engineering, Central South University, Changsha, China, 410083. D. Zhang and J. Ren are also visiting scholars at the University of Waterloo now.(e-mail:\{zdy876,~ren\_ju\}@csu.edu.cn)}%
\thanks{Z. Chen is with the School of Software, Central South University, Changsha, China, 410083, and Z. Chen is the corresponding author (e-mail: czg@mail.csu.edu.cn).}
\thanks{M. K. Awad is with the Computer Engineering Department at Kuwait University, Kuwait City, Kuwait (e-mail: mohamad@ieee.org).}
\thanks{N. Zhang, H. Zhou and X. (S.) Shen are with the Department of Electrical and Computer Engineering, University of Waterloo, Canada, N2L 3G1 (e-mail: \{n35zhang, h53zhou, xshen\}@uwaterloo.ca).}
}
}

\maketitle
%
%
\begin{abstract}
The incorporation of Cognitive Radio (CR) and Energy Harvesting (EH) capabilities in wireless sensor networks enables spectrum and energy efficient heterogeneous cognitive radio sensor networks (HCRSNs). The new networking paradigm of HCRSNs consists of EH-enabled spectrum sensors and battery-powered data sensors. Spectrum sensors can cooperatively scan the licensed spectrum for available channels, while data sensors monitor an area of interest and transmit sensed data to the sink over those channels. In this work, we propose a resource allocation solution for the HCRSN to achieve the sustainability of spectrum sensors and conserve energy of data sensors. The proposed solution is achieved by two algorithms that operate in tandem, a spectrum sensor scheduling algorithm and a data sensor resource allocation algorithm. The spectrum sensor scheduling algorithm allocates channels to spectrum sensors such that the average detected available time for the channels is maximized, while the EH dynamics are considered and PU transmissions are protected. The data sensor resource allocation algorithm allocates the transmission time, power and channels such that the energy consumption of the data sensors is minimized. Extensive simulation results demonstrate that the energy consumption of the data sensors can be significantly reduced while maintaining the sustainability of the spectrum sensors.
\end{abstract}
\begin{IEEEkeywords}
Wireless sensor network, energy harvesting, cognitive radio, energy efficiency, multiple channels
\end{IEEEkeywords}
\IEEEpeerreviewmaketitle
\section{Introduction}
Wireless sensor networks (WSNs) have become a prevalent solution to a wide range of applications including environmental monitoring, patient monitoring and smart homes \cite{Borges2014}. Typically, WSN uses the unlicensed Industrial, Scientific, and Medical (ISM) band for data transmission. However, with the exponential growth in the number of wireless devices operating in this band, WSNs suffer from severe interference \cite{Akan2009}. Cognitive Radio (CR) has emerged as a promising technology to allow secondary unlicensed users to opportunistically access the underutilized spectrum that is licensed to the primary users (PUs) \cite{Zhang2014}. Therefore, CR can reduce the interference and improve spectrum utilization. The integration of CR functions into WSNs leads to Cognitive Radio Sensor Network (CRSN).

In CRSN, spectrum sensors frequently scan the spectrum to obtain higher-resolution estimates of the spectrum availability and guarantee PU protection against interference \cite{Chang2012}. However, this frequent scanning increases the energy consumption of an energy-constrained network, which traditionally operates powered by batteries. Consequently, energy conservation becomes a critical design issue for CRSNs \cite{Ren2015, Ren2016, Su2016}. Energy harvesting (EH) is considered as one of the effective approaches for improving the energy efficiency of WSNs. EH-enabled sensors can harvest energy from either radio signals or ambient energy sources which enable them to operate continuously without battery replacement \cite{Zhang2015a}. In the literature, extensive research efforts have been devoted to improving the energy efficiency of CRSNs. Energy-efficient cooperative spectrum sensing is investigated in \cite{Shah2015, Deng2012}. Shah et al. limit the number of sensors that perform spectrum sensing to minimize the energy consumption by exploiting the spatial correlation of the sensors \cite{Shah2015}. Deng et al. investigate the network lifetime extension of dedicated sensor networks for spectrum sensing \cite{Deng2012}. Despite the importance of these efforts, limits remain for improving the energy efficiency of battery-powered data sensors with low data sensing and limited data transmission rates, for two reasons. First, unlike data sensors, spectrum sensors perform spectrum scanning at a much higher rate than data sensing which depletes the battery energy much faster than the data sensors. Second, harvested energy is sporadic and unstable, whereas battery-stored energy is static and stable, which makes schemes that are developed for battery-powered sensors inapplicable for EH-enabled sensors.

In addition to inefficient spectrum and energy utilization, inaccurate spectrum sensing is another limitation of traditional sensor networks. The spectrum-scanning results of a single spectrum sensor are prone to detection error due to the spatially large-scale effect of shadowing and small-scale effect of multipath fading \cite{Akyildiz2011}.  Alternatively, cooperative spectrum sensing can be performed to enhance the accuracy of spectrum sensing \cite{Zhang2014a}. In cooperative spectrum sensing, multiple spectrum sensors sense the same channel and coordinate their decisions on the availability of a given channel. Hence, the incorporation of energy harvesting and cognitive radio techniques in addition to cooperative spectrum sensing brings major improvements to traditional WSNs. Energy-efficient cooperative spectrum sensing has been the focus of several research activities. In \cite{Cheng2012}, Cheng et al. schedule a group of spectrum sensors between the active and inactive states to improve the performance of spectrum sensing. Considering the impact of frequent state switching on sensors' stability, the schedule minimizes the sensors' switching frequency among the states. In \cite{Zhang2015}, Zhang et al. design a distributed cooperative spectrum sensing scheme, wherein the Secondary Users (SUs) only exchange their measurements with the one-hop neighbours. In \cite{Khan2010}, Khan et al. propose a selection scheme to find the sensors with the best detection performance for cooperative spectrum sensing, without requiring a priori knowledge of the primary-user-signal-to-noise ratio (SNR). In \cite{Liu2013a}, Liu et al. propose an ant colony-based algorithm for a dedicated sensor network, whereby spectrum sensing is performed to support the operation of a secondary network. Throughput of the secondary network is optimized by scheduling the spectrum-sensing activities according to the residual energy of each sensor. Additionally, to achieve energy efficient cooperative spectrum sensing, parameters are optimized, such as the detection threshold~\cite{Ebrahimzadeh2015}, sensing duration~\cite{Eryigit2013}, and switch cost~\cite{Bayhan2013}. EH-aided cooperative spectrum sensing should take the EH dynamics of the sensors into consideration. For EH-aided spectrum sensing, the objective is to explore as many of the available channels as possible while maintaining the sustainability of the spectrum sensors, considering the diverse energy-harvesting capabilities of the spectrum sensors. However, the aforementioned energy-efficient cooperative spectrum-sensing schemes which only focus on the minimization of the energy consumption of spectrum sensing, cannot be directly applied to EH-aided cooperative spectrum sensing.

In this paper, we propose a resource allocation solution to address the gaps that are identified above in the existing works, namely, spectrum under-utilization, energy inefficiency and spectrum-sensing inaccuracy. Specifically, for a heterogeneous cognitive radio sensor network (HCRSN) that is composed of EH-enabled spectrum sensors and battery-powered data sensors, we develop a solution that can jointly guarantee the sustainability of spectrum sensors, the energy efficiency of the data sensors and the accuracy of spectrum sensing. The HCRSN operates over two phases, i.e., a spectrum-sensing phase followed by a data transmission phase. In the spectrum-sensing phase, EH-enabled spectrum sensors cooperatively sense the spectrum to detect underutilized channels that are licensed to the primary network. Spectrum-sensing scheduling is optimized to maximize the detected channel's available time considering the dynamics of EH. In the data transmission phase, the available channels are utilized by data sensors for sensed data transmission. We combine the resource management and allocation of each phase in a unified solution. Despite the physical independence of spectrum sensors and data sensors, a unified solution is necessary to optimize the overall energy efficiency and performance of HCRSNs. The imbalance of energy replenishment and consumption at either the spectrum or data sensors results in nodes failure and deteriorates the network performance; thus, energy should be managed under one unified setup. Furthermore, the performance of spectrum sensor scheduling in the first phase highly impacts the energy efficiency of data sensors in the second phase. A longer channel available time detected in the first phase increases the channel access time and decreases the probability of collision in the second phase. This causal impact of spectrum sensor scheduling performance in the first phase on the performance of data sensor resource allocation in the second phase necessitates a unified solution. On the other hand, it is practically infeasible to allocate channels, transmission time, and transmission power before the spectrum sensors identifies the available channels. Therefore, spectrum sensor scheduling and data sensor resource allocation have to be addressed over two coupled problems operating in tandem but under one unified setup. Summarily, the contributions of this paper are twofold:
\begin{enumerate}
\item We formulate the EH-aided spectrum-sensing problem as a nonlinear integer programming problem and propose a Cross-Entropy-based algorithm to maximize the average available time of the channel under the protection for PUs.
\item We propose a joint time and power allocation algorithm to minimize the energy consumption of the data sensors, based on the analysis of the channel fading and the exponential \textit{ON-OFF} model of the PUs' behavior.
\end{enumerate}

It is imperative to mention here that the literature schemes, which consider only energy minimization of spectrum sensors \cite{Shah2015,Deng2012, Khan2010, Shah2014}  rather than channel available time maximization are inapplicable in HCRSN. Furthermore, unlike existing solutions that separately consider channels allocation \cite{Byun2008} \cite{Hu2011} and power control \cite{AyalaSolares2012} \cite{Naeem2013}, the proposed solution jointly allocates time, frequency and power to data sensors; hence, improve the energy efficiency of data sensors.

The remainder of this paper is organized as follows. The network architecture and cognitive radio model are detailed in Section \ref{sec:system_model}. A mathematical formulation and the proposed solutions of the spectrum sensor scheduling problem and data sensors resource allocation problem are detailed in Section \ref{Sec:problem_and_solution}. Performance evaluation results that demonstrate the efficiency of the proposed algorithms are presented in Section \ref{sec:evaluation}. Conclusions are drawn in Section \ref{sec:conlusion}.
\section{System Model} \label{sec:system_model}

\subsection{Network Architecture}
\begin{figure}[ht!]
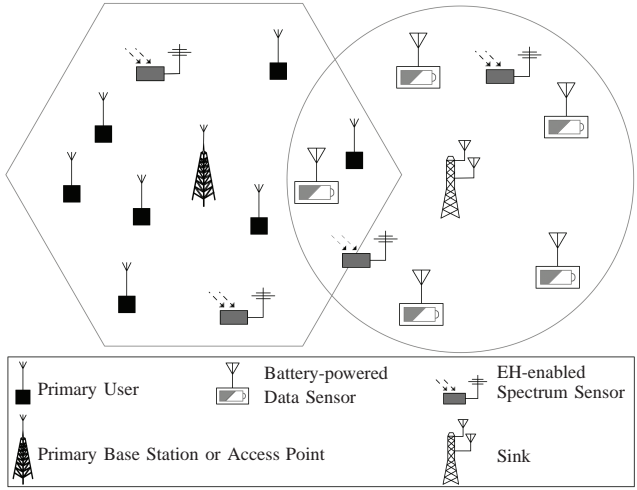

\begin{center}
\begin{lpic}{Figure_1(0.43,0.43)} 
\lbl[l]{10,30; \scriptsize Primary User}
\lbl[l]{10,10; \scriptsize Primary Base Station or Access Point}
\lbl[l]{80,35; \scriptsize Battery-powered}
\lbl[l]{80,28; \scriptsize Data Sensor}
\lbl[l]{152,10; \scriptsize Sink}
\lbl[l]{152,36; \scriptsize EH-enabled}
\lbl[l]{152,29; \scriptsize Spectrum Sensor}
\end{lpic}
\end{center}
\caption{An illustration of heterogeneous cognitive radio sensor network. \label{fig_netm}}
\end{figure}
We consider a HCRSN that consists of three types of nodes: $N$ battery-powered data sensors, $M$ EH-enabled spectrum sensors and a sink node, as shown in Fig. \ref{fig_netm}.  The HCRSN coexists with a network of PUs that have access to the licensed spectrum. The licensed spectrum is divided into $K$ orthogonal channels that have equal bandwidth $W$. Spectrum sensors are deployed to sense and identify available channels that are not utilized by PUs, whereas data sensors collect data from an area of interest. The data is then transmitted over the available channels to the sink.

The considered HCRSN operates as follows: First, the sink schedules spectrum sensors to detect the PUs' presence over channels using energy detection \cite{Zhang2014}. A PU is determined to be active, i.e., channel unavailable, if at least one scheduled spectrum sensor reports it to be present on the channel \cite{Deng2012}. The energy consumption of a spectrum sensor used to detect one channel is determined by $e_s = \tau_s \cdot P_s$, where $P_s$ is the power consumption of the spectrum sensing. We assume that the EH rate is known a priori and is stable over $T$ \cite{Zhang2013}. To guarantee the sustainability of the $m$-th spectrum sensor, its energy consumption should not exceed the amount of harvested energy in one period, $\pi_m \cdot T$, where $\pi_m$ denotes the average EH rate of the $m$-th spectrum sensor. Second, the sink assigns the available channels to the data sensors for data transmission.

\begin{table}[!t]
    \caption{The Key Notations}
    \centering
    \small
    \begin{tabular}{p{1.5cm}|p{6.6cm}}
         \hline
         \hline 
         \textbf{Notation} & \textbf{Definition} \\
         \hline
         \hline
         $M$ & Number of energy harvesting spectrum sensors \\
         $N$ & Number of data sensors\\
         $K$ & Number of licensed channels\\
         $U$ & Number of samples for detecting one channel\\
         $B$ & Number of transceivers mounted on the sink \\
         $F_f, F_d$ & Probability of cooperative false alarm and detection\\
         $\bar{\gamma}_{m,k}$ & Signal to Noise Ratio (SNR) from $PU_k$ to EH spectrum sensor $m$  \\
         $T$ & Time length of each period\\
         $\tau_0$ & Time length of spectrum sensing phase\\
         $\pi_m$ & Energy harvesting rate of EH spectrum sensor $m$ \\
         $\tau$ & Time length of per spectrum sensing  \\
         $P_s$ & power for spectrum sensing\\
         $e_s$ & Energy consumption per spectrum sensing \\
         $\lambda_k$ &  Transition rate of PU from state ON to state OFF\\
         $\mu_k$ &  Transition rate of PU from state OFF to state ON \\
         $P_{ON}^k, P_{OFF}^k$ &  Stationary probability of PU on channel $k$ is in state ON and OFF \\
         $L_{ON}^k, L_{OFF}^k$ &  Sojourn time of PU on channel $k$ in state ON and OFF   \\
         $D_n$ & Required amount of data transmitted to sink\\
         $h_{n,k}$ & Channel fading of data sensor $n$ through channel $k$\\
         $d_n$ & Distance from data sensor $n$ to the sink\\
         $W$ &  Channel band in Hz\\
         $\eta_k$ & Collision probability on channel $k$ \\
         $t_{n,k}$ &  transmission time of data sensor $n$ on channel $k$\\
         $P_{n,k}$ &  transmission power of data sensor $n$ on channel $k$\\
         $\overline{MD}_d$ & Threshold of mis-detection probability \\
         \hline
    \end{tabular}
    \label{tab_symbol}
\end{table}
\begin{figure}[ht!]
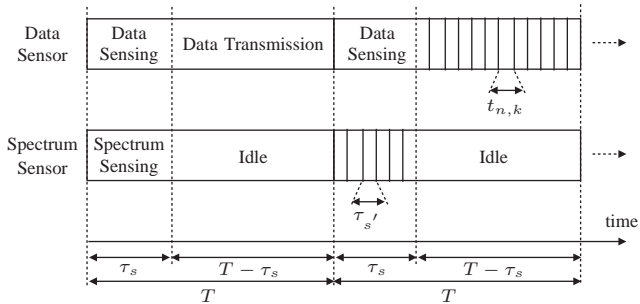

\begin{center}
\begin{lpic}{Figure_2(1,1)} 
\lbl[.]{13.5,36; \scriptsize Data}
\lbl[.]{13.5,33; \scriptsize Sensing}
\lbl[.]{30,34.5; \scriptsize Data Transmission}
\lbl[.]{13.5,21; \scriptsize Spectrum}
\lbl[.]{13.5,18; \scriptsize Sensing}
\lbl[.]{30,19.5; \scriptsize Idle}
\lbl[.]{13.5,4.5; \scriptsize $\tau_s$}
\lbl[.]{29.5,4.5; \scriptsize $T-\tau_s$}
\lbl[.]{24,1; \scriptsize $T$}
\lbl[.]{46.5,36; \scriptsize Data}
\lbl[.]{46.5,33; \scriptsize Sensing}
\lbl[.]{63.5,26; \scriptsize $t_{n,k}$}
\lbl[.]{45.2,11; \scriptsize $\tau_{s^{'}}$}
\lbl[.]{62,19.5; \scriptsize Idle}
\lbl[.]{46.5,4.5; \scriptsize $\tau_s$}
\lbl[.]{62,4.5; \scriptsize $T-\tau_s$}
\lbl[.]{56,1; \scriptsize $T$}
\lbl[.]{2,21; \scriptsize Spectrum}
\lbl[.]{2,18; \scriptsize Sensor}
\lbl[.]{2,36; \scriptsize Data}
\lbl[.]{2,33; \scriptsize Sensor}
\lbl[.]{79,11; \scriptsize time}
\end{lpic}
\end{center}
\caption{Timing diagram and frame structure of the HCRSN.  \label{fig_tml}}
\end{figure}


Fig. \ref{fig_tml} shows the timing diagram and frame structure of the considered network. The HCRSN operates periodically over time slots of duration $T$. Each time slot is divided into two phases: the spectrum sensing phase and data transmission phase. In the spectrum sensing phase, the spectrum sensors cooperatively identify the presence of PUs, while the data sensors collect information from the area of interest. The duration of the spectrum sensing phase is $\tau_s$, which is further divided into mini-slots of duration $\tau_{s^{'}}$ over which a single spectrum sensor senses one channel.  After the spectrum sensing phase, the sink collects the results from all the scheduled spectrum sensors and estimates the availability of the channels. Then, the sink optimizes the data transmission scheduling of data sensors to conserve their energy. The data sensors transmit data according to the schedule in the subsequent data transmission phase with duration $T - \tau_s$ divided over the time slots of duration $t_{n,k}$ in which the $n$-th data sensor transmits data to the sink over the $k$-th channel.

With respect to the notation, the following holds: a bold-face small-case symbol always refers to a vector; and a non-italic bold-face large-case symbol always symbolizes a matrix.
\subsection{Cognitive Radio Model}
\label{subsec_cr}
All of the channels experience slow and flat Rayleigh fading with similar fading characteristics. The PU behavior over each channel is modeled as a stationary exponential \textit{ON-OFF} random process, in which the \textit{ON/Active} and \textit{OFF/Inactive} states represent the presence and absence of a PU over a channel, respectively . We use $\lambda_k$ to denote the transition rate from the state \textit{Active} to the state \textit{Inactive} on the $k$-th channel and $\mu_k$ to denote the transition rate in the reverse direction. The estimation of  $\lambda_k$ and $\mu_k$ is out of the scope of this work; however, they can be obtained by the channel parameter estimation schemes, similar to the ones proposed in \cite{Tehrani2012} and \cite{Kim2008}.  The channel usage changes from one PU to the other and, hence, affects the transition rates.


Spectrum sensors perform binary hypothesis testing to detect the presence of PU signals over channels. Hypothesis 0 ($\mathcal{H}_0$) proposes that the PU is \textit{Inactive} and the channel is available, while Hypothesis 1 ($\mathcal{H}_1$) proposes that the PU is \textit{Active} and the channel is unavailable. The spectrum sensor receives a sampled version of the PU signal. The number of samples is given by $U=\tau_s f_s$, where $f_s$ is the sampling frequency. An energy detector is applied to measure the energy that is associated with the received signal. The output of the energy detector, i.e., the test statistic, is compared to the detection threshold $\varepsilon$, to make a decision on the state of the PU, \textit{Active} or \textit{Inactive}. The test statistic evaluates to $Y_{m,k} = \frac{1}{U} \sum_{u=1}^U |y_{m,k}(u)|^2$, where $y_{m,k}(u)$ is the $u$-th sample of the received signal at the $m$-th spectrum sensor on the $k$-th channel. We assume that the PU signal is a complex-valued PSK signal and the noise is circularly symmetric complex Gaussian with zero mean and $\sigma^2$ variance \cite{Liang2008}.

The performance of the energy detector is evaluated by the the following performance metrics under hypothesis testing \cite{Atapattu:2014:}:
\begin{itemize}
\item The false alarm probability $p_f (m,k)$: The probability that the $m$-th spectrum sensor detects a PU to be present on the $k$-th channel when it is not present in fact, i.e., $\mathcal{H}_0$ is true. The false alarm probability is given by \cite{Liang2008},
\begin{equation}
\label{pf}
p_f (m,k) = Pr(Y_{m,k} > \varepsilon | \mathcal{H}_0) = Q\left( \left(\frac{\varepsilon}{\sigma^2} - 1 \right) \sqrt{U} \right),
\end{equation}
where $Q(\cdot)$ is the complementary distribution function of the standard Gaussian. Without loss of generality, we set the detection threshold to be the same for all of the spectrum sensors; hence, the false alarm probability becomes fixed for all of the sensors and is denoted by $\bar{p}_f$.
\item  The detection probability $p_d (m,k)$: The probability that the $m$-th spectrum sensor detects the presence of a PU on the $k$-th channel while $\mathcal{H}_1$ is true. This probability was found to be \cite{Liang2008}
\begin{equation}
\small
\label{eqn_pd}
p_d(m, k) = Pr(Y_{m,k} > \varepsilon | \mathcal{H}_1) =Q\left(\frac{Q^{-1}(\bar{p}_f) - \sqrt{U}\gamma_{m,k}}{\sqrt{2 \gamma_{m,k} + 1}} \right),
\end{equation}
\end{itemize}
where $\gamma_{m,k}$ denotes the received signal-to-noise ratio (SNR) from the PU on the $k$-th channel. To reduce the communication overhead and delay, each spectrum sensor sends the final 1-bit decision (e.g., 0 or 1 represents the \textit{Active} or \textit{Inactive} state, respectively) to the sink. The sink makes the final decision on the presence of a PU following the Logic-OR rule \cite{Liang2008,Deng2012}. Under this rule, the PU is considered to be present if at least one of the scheduled sensors reports that it is present. Therefore, the final false alarm probability $F_f^k$ and final detection probability $F_d^k$ can be written as
\begin{eqnarray}
F_f^k  & = & 1- \Pi_{m \in \mathcal{M}_k} (1 - \bar{p}_f) \label{eqn_ff},~\mbox{and}\\
F_d^k & = & 1- \Pi_{m \in \mathcal{M}_k} (1 - p_d(m, k)) \label{eqn_fd},
\end{eqnarray}
where $\mathcal{M}_k$ represents the set of spectrum sensors that is scheduled to detect the $k$-th channel.

\section{Problem Statement and Proposed Solutions\label{Sec:problem_and_solution}}
\begin{figure*}[ht]
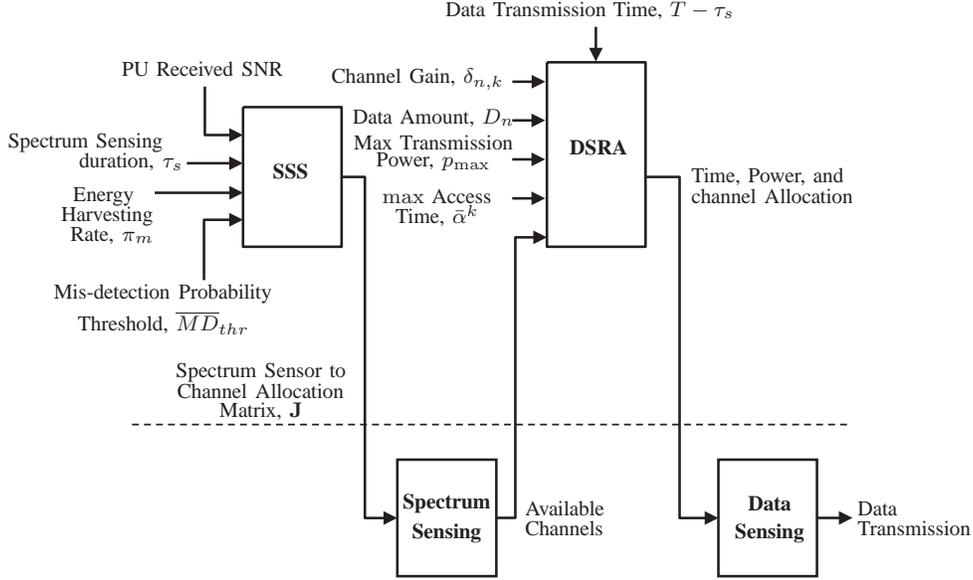

\begin{center}
\normalsize
\setcounter{lpgridstep}{1}
\setcounter{lpcoordstep}{5}
\begin{lpic}[t(0mm),b(0mm)]{Figure_3(1.3,1.3)} 
\lbl[.]{18,44.5; \footnotesize \textbf{SSS}}
\lbl[.]{49.5,47; \footnotesize \textbf{DSRA}}
\lbl[.]{34,10.5; \footnotesize \textbf{Spectrum}}
\lbl[.]{34,7.5; \footnotesize \textbf{Sensing}}
\lbl[.]{67,10.5; \footnotesize \textbf{Data}}
\lbl[.]{67,7.5; \footnotesize \textbf{Sensing}}
\lbl[.]{31,54; \footnotesize Channel Gain, $\footnotesize \delta_{n,k}$}
\lbl[.]{33,50; \footnotesize Data Amount, $D_n$ }
\lbl[.]{33,47.5; \footnotesize Max Transmission } 
\lbl[.]{33,45.5; \footnotesize Power, $p_{\max}$ }
\lbl[.]{33,42; \footnotesize $\max$ Access}
\lbl[.]{33,40; \footnotesize Time, $\bar{\alpha}^{k}$}
\lbl[.]{49,61; \footnotesize Data Transmission Time, $T-\tau_s$ }
\lbl[l]{59,44; \footnotesize Time, Power, and }
\lbl[l]{59,42; \footnotesize channel Allocation}
\lbl[l]{76,10; \footnotesize Data}
\lbl[l]{76,8; \footnotesize Transmission}
\lbl[l]{42,10; \footnotesize Available}
\lbl[l]{42,8; \footnotesize Channels}
\lbl[.]{9,55; \footnotesize PU Received SNR}
\lbl[r]{5,47.5; \footnotesize Spectrum Sensing}
\lbl[r]{7,45.5; \footnotesize duration, $\tau_s$}
\lbl[r]{2,42; \footnotesize Energy}
\lbl[r]{4,40; \footnotesize Harvesting}
\lbl[r]{4,38; \footnotesize Rate, $\pi_m$}
\lbl[.]{5,32; \footnotesize Mis-detection Probability}
\lbl[.]{5,29; \footnotesize Threshold, $\overline{MD}_{thr}$}
\lbl[.]{15,24; \footnotesize Spectrum Sensor to}
\lbl[.]{15,22; \footnotesize Channel Allocation}
\lbl[.]{15,20; \footnotesize Matrix, ${\bf J}$}
\end{lpic}
\end{center}

\caption{A block diagram of the proposed system. The dashed line separates the optimization plane from the sensing plane. \label{fig_flowchart}}
\end{figure*}

, which shows the scheduling and resource allocation problems, the spectrum-sensing and data-sensing phases, and the data flows among them. The dashed line seperates the optimization plane from the sensing plane

In an HCRSN with the above-described architecture, cognitive radio models and EH dynamics, the problems of scheduling the spectrum sensors and allocating the resources for the data sensors become challenging. In the first problem, the spectrum sensor scheduling (SSS) problem, the sink schedules the spectrum sensors to sense the presence of the PUs over the channels in such a way that the channel availability is maximized while respecting the EH dynamics and PUs' priorities in accessing the channels. Solving this problem makes the available channels known to the sink which allocates them to the battery-powered data sensors along with the transmission time and power allocation, with the objective of minimizing the data sensors' energy consumption. This resource allocation problem is referred to as the data sensor resource allocation (DSRA) problem.

Fig. \ref{fig_flowchart} shows the two problems in tandem, the spectrum-sensing and data-sensing phases, and the data flows among them. In the following two subsections, we present problem formulations and solutions for both problems. The first problem is formulated as a nonlinear integer programming problem, while the second problem is formulated as a biconvex optimization problem.

\subsection{Spectrum Sensors Scheduling} \label{sec:spectrum_sensing}
In this subsection, we investigate the SSS problem which is posed as a nonlinear integer programming problem. Through a Cross-Entropy-based solution, the channel availability is maximized while guaranteeing EH-enabled spectrum sensors sustainability and PUs protection.
\subsubsection{Problem Formulation}
Three factors impact the average detected available time of the channel: the actual average available time, the final false alarm probability complement ($1 - F_f^k$) and the final detection probability complement ($1 - F_d^k$). The actual average available time of the $k$-th channel is the product of the mean sojourn time and the stationary probability of the $k$-th channel. Let $\bar{L}_{\mbox{\textit{\footnotesize Active}}}^k= \frac{1}{\lambda_k}$ and $\bar{L}_{\mbox{\textit{\footnotesize Inactive}}}^k = \frac{1}{\mu_k}$ denote the mean sojourn time of the \textit{Active} state and the \textit{Inactive} state on the $k$-th channel, respectively. Moreover, the stationary probabilities of the \textit{Active} and \textit{Inactive} states are given by
\begin{equation}\label{eqn_Poo}
P_{Active}^k = \frac{\mu_k}{\lambda_k + \mu_k}, \quad P_{\mbox{\footnotesize \textit{Inactive}}}^k = \frac{\lambda_k}{\lambda_k + \mu_k}.
\end{equation}
Therefore, the $k$-th channel average actual available time is given by,
\begin{equation}
\alpha^k = \bar{L}_{\mbox{\textit{\footnotesize Inactive}}}^k \cdot P_{\mbox{\textit{\footnotesize Inactive}}}^k.
\end{equation}
Let $\bf{J}$ be an $M \times K$ matrix with binary elements $[{\bf J}]_{m,k}$. A binary element of $1$ indicates the assignment of the $m$-th spectrum sensor to detect the $k$-th channel and $0$ otherwise. Given that the PU on the $k$-th channel is inactive, the probability that the $k$-th channel is available is equivalent to the complement of the final false alarm probability, which can be written as
\begin{equation}
1-F_f^k = \Pi_{m \in \mathcal{M}_k} (1 - \bar{p}_f) = (1-\bar{p}_f)^{\sum_{m=1}^M [{\bf{J}}]_{m,k}}
\end{equation}
The data sensor transmission interferes with the PU transmission if the spectrum sensors do not detect the PU presence while it is present. The chance of this event is captured by the mis-detection probability $1 - F_d^k$. To protect the PU from such interference, we consider detection decisions with a mis-detection probability of less than $\overline{MD}_{thr}$. A binary variable $I^k_{d}$ is introduced to indicate whether the protection requirement is satisfied or not and is given by,
\begin{equation}
\label{eqn_id}
I_d^k= \left\{
\begin{array}{rl}
1, & \text{if }  1 - F_d^k < \overline{MD}_{thr},\\
0, & \text{otherwise}.
\end{array} \right.
\end{equation}
If the mis-detection probability of the $k$-th channel exceeds  $\overline{MD}_{thr}$, the detection is considered to be unreliable, and the $k$-th channel is not accessed by data sensors. Substituting Eqn. (\ref{eqn_pd}), (\ref{eqn_fd}) into (\ref{eqn_id}) yields,
\begin{equation} \label{eqn_ind2}
I_d^k= \left\{
\begin{array}{rl}
1, & \text{if }   \Pi_{m \in \mathcal{M}_k} \left( 1- Q\left(\frac{Q^{-1}(\bar{p}_f) - \sqrt{U}\gamma_{m,k}}{\sqrt{2 \gamma_{m,k} + 1}} \right) \right)  \\
 &< \overline{MD}_{thr},\\ %
0,& \text{otherwise}.
\end{array} \right.
\end{equation}

The objective function of SSS that maximizes the average detected available time of a channel while protecting the PU can be written as follows:
\begin{equation}
\label{SSS-objective}
\sum_{k=1}^K \alpha^k (1-\bar{p}_f)^{\sum_{m=1}^M [{\bf J}]_{m,k}} I_d^k.
\end{equation}
The SSS is subject to two constraints; the first constraint is related to the EH dynamics, whereas the second constraint is related to the frame structure (see Fig. \ref{fig_tml}). In a given frame $T$, to maintain the sustainability of the spectrum sensors, the energy consumption of each sensor should not exceed its harvested energy. This arrangement can be mathematically written as
\begin{equation} \label{CE_power}
(\sum_{k=1}^{K} [{\bf{J}}]_{m,k}) e_s \leq \pi_m T.
\end{equation}
Moreover, the time that is used for sensing the $k$-th channel is bounded by the duration of the spectrum-sensing phase $\tau_s$ in one period, namely,
\begin{equation} \label{CE_time}
(\sum_{m=1}^{M} [{\bf J}]_{m,k}) \tau_{s^{'}} \leq \tau_s.
\end{equation}
Then, the spectrum sensor scheduling problem becomes a combinatorial problem of optimizing the sensor-to-channel assignment matrix $\bf{J}$ and can be written as follows:
\begin{align*}
\label{SSS-Problem}
&(\mbox{SSS})~ \max_{\bf{J}} \sum_{k=1}^K \alpha^k (1-\bar{p}_f)^{\sum_{m=1}^M [{\bf J}]_{m,k}} I_d^k    \\ \nonumber 
{\rm s.t.}~ &\begin{cases}
(\sum_{k=1}^{K} [{\bf J}]_{m,k}) e_s \leq \pi_m T,  \forall m,\\
(\sum_{m=1}^{M} [{\bf J}]_{m,k}) \tau_{s^{'}} \leq \tau_s, \forall m,\\
  [{\bf J}]_{m, k}= \{0 , 1\} ~~\forall m, k.
\end{cases}
\end{align*}
The term $\alpha^k$ has a constant value over a given channel. As more channels are assigned to a given spectrum sensor, i.e., as ${\sum_{m=1}^M [{\bf J}]_{m,k}}$ increases, the value of $(1-\bar{p}_f)^{\sum_{m=1}^M [{\bf J}]_{m,k}}$ decreases, and $I_d^k$ tends to take a unit value. Therefore, there exists a trade-off between $(1-\bar{p}_f)^{\sum_{m=1}^M [{\bf J}]_{m,k}}$ and $I_d^k$. However, the assignment $[{\bf J}]_{m,k}$ exists in the exponential part of $(1-\bar{p}_f)^{\sum_{m=1}^M {\bf J}_{m,k}}$ and affects $I_d^k$ through $\mathcal{M}_k$. These structures make the SSS an integer programming problem. Intuitively, the objective function in Eqn. (\ref{SSS-objective}) can be optimized by performing an exhaustive search over the space that is characterized by the constraints of SSS. However, this arrangement leads to a search space of size $2^{MK}$ which is computationally prohibitive especially for the resource-limited sensor network. In the following subsection, we apply the Cross-Entropy-based algorithm (C-E algorithm) \cite{Rubinstein1999} to address (SSS). Although the performance bound of the C-E algorithm remains an open theoretical issue \cite{DeBoer2005}, it has been shown effective in solving a similar combinatorial optimization problem \cite{Zhang2014}.

\subsubsection{Cross Entropy-based Algorithm}
The basic idea of the C-E algorithm lies in the transformation of a deterministic problem into the related stochastic optimization problem such that rare-event simulation techniques can be applied. More specifically, an associated stochastic problem (ASP) is defined for the deterministic problem, and then, the ASP problem is solved using an adaptive scheme. The adaptive scheme generates random solutions that converges stochastically to the optimal or near-optimal solution of the original deterministic problem.

Before introducing the C-E algorithm, we transform the constrained problem into an unconstrained problem by applying a penalty method. Let $\omega = -\sum_{k=1}^K \alpha^k$ be the penalty for violating any of the constraints, and then, the SSS problem transforms to
\begin{equation}
\label{eqn_objce}
\begin{aligned}
O = & ~\omega \cdot  I_{(\sum_{m=1}^M [{\bf J}]_{m, k} \cdot e_s> \pi_m T)} + \omega \cdot I_{(\sum_{k=1}^K [{\bf J}]_{m, k} \cdot \tau_{s^{'}} > \tau_s)}  \\
& +\sum_{k=1}^K \alpha^k (1-\bar{p}_f)^{\sum_{m=1}^M [{\bf J}]_{m,k}} I^k_d.
\end{aligned}
\end{equation}
For a positive constant penalty of $\omega$, the unconstrained objective function evaluates to a negative value for all of the infeasible solutions that violate constraints (\ref{CE_power}) and (\ref{CE_time}). The indicator function, $I_{(\cdot)}$, takes the value of $1$ for true evaluations of $(\cdot)$ and zero otherwise.

Recall that the sink schedules the spectrum sensors to detect the presence of PU on certain licensed channels. Therefore, the row vectors of ${\bf J}$ are drawn from a set, $\mathcal{C}$, of channel assignment vectors that hold a sequence of binary numbers, $\mathcal{C}=\{\bm{1}, \cdots, \bm{c}, \cdots, \bm{C}\}$, and the cardinality of the set is $C= |\mathcal{C}| = 2^K$. Mathematically, $[{\bf J}]_{m,1:K} \in \mathcal{C}$. Although $C$ grows exponentially with $K$, we focus on a single hop network in which the number of potential channels is limited, e.g., $4 - 6$; hence, the cardinality $C$ is also limited. Next, we allocate channel assignment vectors to the spectrum sensors rather than individual channels as in ${\bf J}$. Define a channel assignment vector to the spectrum sensors binary assignment matrix, ${\bf V}^z=\{v^z_{m,\bm{c}}~|~\ 1 \leq m \leq M, \bm{c} \in \mathcal{C}\}$, of size $M \times C$, where a value of $1$ for $v^z_{m,\bm{c}}$ indicates that the  channel assignment vector $\bm{c}$ is allocated to the $m$-th spectrum sensor. In one of the steps of the C-E algorithm, random samples of this matrix are generated, and the superscript $z$ is introduced to denote the sample number.

The ${\bf V}^z$ samples are generated following a probability mass function (p.m.f) that is denoted by matrix ${\bf Q}^i$, which is defined as ${\bf Q}^i := \{q_{m,\bm{c}}^i~|~\ 1 \leq m \leq M, \bm{c} \in \mathcal{C}\}$, where $q_{m,\bm{c}}^i$ denotes the probability that $m$ is scheduled to sense the channels in vector $\bm{c}$. The C-E algorithm operates iteratively, and in every step, the p.m.f matrix is updated. The superscript $(\cdot)^i$ denotes the iteration number. Each iteration of the C-E algorithm consists of the following steps:
\begin{enumerate}[]
\item \textit{Initialization:} Set the iteration counter to $i = 1$ and the maximum iteration number to $i_{\max}$. Set the initial stochastic policy of all of the spectrum sensors to be a uniform distribution on the channel assignment vector set $\mathcal{C}$, such that $m$ chooses vector $\bm{c}$ with probability $q_{m,\bm{c}}^1 = 1/C, ~\forall m, \bm{c}$.

\item \textit{Generation of Sample Solutions}: Generate $Z$ samples of the matrix ${\bf V}^z$ based on the p.m.f matrix ${\bf Q}^i$. Note that each spectrum sensor is randomly assigned one channel assignment vector that holds several channels, i.e., $\sum_{1}^{C} v^z_{m,\bm{c}} =1, ~ \forall z ~\forall m$. \label{ce_gs}

\item \textit{Performance Evaluation}: Substitute the $Z$ samples of ${\bf V}^z$ into Eqn. (\ref{eqn_objce}) to obtain an objective function value $O^z$ for each sample; the superscript $(\cdot)^z$ has been introduced to denote the sample number. Sort the $Z$ values of $O^z$ in descending order. Set $\rho$ to be a fraction of the sorted objective values to retain, and then, take the largest $\lceil \rho Z \rceil$ values of the sorted set and ignore all of the others. Moreover, set $\eta$ to be the smallest value in the sorted and retained set.

\item \textit{p.m.f. Update}: Update the p.m.f. based on the retained objective function values. The value of $q_{m,\bm{c}}^{i+1}$ is determined by
    \begin{equation}
    \label{C-D_Update}
    q_{m,\bm{c}}^{i+1} = \frac{\sum_{z=1}^Z v^z_{m,\bm{c}} I_{O^z \geq \eta}}{\lceil \rho Z \rceil},
    \end{equation}
In this step, the channel vector assignment probability $q_{m,\bm{c}}^i$ is updated by increasing the probability of assignments that are generating large objective function values over the various randomly generated samples.

\item \textit{Stopping Criterion}: The algorithm stops iterating if the maximum number of iterations $i_{\max}$ is reached or the following inequality stands
\begin{equation}
    \label{eqn_threshold}
||{\bf Q}^{i+1} - {\bf Q}^i ||_{Fr} \leq \epsilon,
\end{equation}
where $||\cdot ||_{Fr}$ denotes the Frobenius norm\footnote{The Frobenius norm is defined as the square root of the sum of the absolute squares of the elements of the matrix. For example, if
\begin{equation*}
\bm{A} =
\begin{bmatrix}
a_{11} & a_{12} \\
a_{21} & a_{22}
\end{bmatrix},
\end{equation*}
then
\begin{equation*}
|| \bm{A} ||_{Fr} = \sqrt{|a_{11}|^2 + |a_{12}|^2 + |a_{21}|^2 + |a_{22}|^2}.
\end{equation*}
}. Otherwise, increment the iteration counter $i$ and go back to Step \ref{ce_gs}. Eqn. (\ref{eqn_threshold}) represents the convergence condition of p.m.f ${\bf Q}^i$. It was shown in \cite{Costa2007} that the sequence of p.m.f converges with probability 1 to a unit mass that is located at one of the samples.

Note that fine tuning the values $\epsilon$ and $i_{\max}$ impacts the convergence speed of the algorithm and the quality of the obtained solution. A large value of $\epsilon$ results in faster convergence but a shorter average available time of the channel. Additionally, a larger value for $i_{\max}$ leads to a slower convergence speed but also leads to a longer average available time for the channel.
\item \textit{Solution Selection}: When the algorithm terminates, select the solution ${\bf V}^z$ that generates the largest objective value $O^z$. Set the values of ${\bf J}$ based on the assignments solution in ${\bf V}^z$. In other words, the channel-vector-to-the-spectrum-sensors assignment in ${\bf V}^z$ is mapped to the channels-to-spectrum-sensors assignment in ${\bf J}$ which is a solution to the original problem SSS.
    \label{ce_output}
\end{enumerate}

The sink schedules spectrum sensors to detect the licensed channels according to the solution obtained in Step-\ref{ce_output}. After the spectrum-sensing phase, spectrum sensors report their decisions on the channel availability to the sink. The sink estimates the availability of each channel based on the Logic-OR rule and utilizes the available channels to collect data from the data sensors. In the following, we investigate the data sensor resource allocation (DSRA) problem.

\subsection{Transmission Time and Power Allocation in the Data Transmission Phase}
\label{sec:data_transmission}
In the data transmission phase, data sensors report the collected data to the sink. Because a data sensor is battery-powered, minimizing its energy consumption becomes critical to prolong its lifetime. To accomplish this goal, we first formulate the problem of the data sensors' transmission time and power allocation as a biconvex optimization problem, and then, we propose a joint time and power allocation (JTPA) algorithm to obtain a solution.
\subsubsection{Problem Formulation}
Available channels detected by the spectrum sensors are allocated to the $B$ cognitive radio transceivers that are mounted on the sink. If the number of available channels is less than $B$, then all of the available channels are allocated. Alternatively, the available channels are sorted with respect to their sojourn time, and the channels with the largest sojourn time values are allocated to transceivers. Let $\bar{K}$ be the number of allocated channels, and note that $\bar{K} \leq B$. Because all of the channels have the same bandwidth and average power gain, a long average sojourn time implies a large capacity.

Recall that $\alpha^k$ is the $k$-th channel's available time. However, scheduling the data sensors to transmit for the entire $\alpha^k$ increases the chance of collision between the data sensor and the returning PU. Let $\bar{\alpha}^{k}$ be the maximum access time of the $k$-th channel, where $\bar{\alpha}^{k} < \alpha^k$. It is important to design $\bar{\alpha}^{k}$ such that a low collision probability $p_{coll}^k(\bar{\alpha}^{k})$ is maintained on the $k$-th channel. Given that the PU behavior on each channel is a stationary exponential ON-OFF random process, the probability of collision $p_{coll}^k(\bar{\alpha}^{k})$ is given by \cite{Kim2008}

\begin{equation} \label{eqn_col_prob}
p_{coll}^k(\bar{\alpha}^{k}) = P_{\mbox{\footnotesize \textit{Inactive}}}^k \cdot (1-e^{-\mu_k \bar{\alpha}^{k}}).
\end{equation}
where $P_{\mbox{\footnotesize \textit{Inactive}}}^k$ is the probability that PU is not present on the $k$-th channel at the beginning of the data transmission phase, and $(1-e^{-\mu_k \bar{\alpha}^{k}})$ captures the probability that PU returns in $[0, \bar{\alpha}^{k}]$. The detailed derivation of Eqn. (\ref{eqn_col_prob}) is provided in Appendix \ref{Apn_col_prob}.


To maintain a target collision probability $\overline{p_{coll}^k}$, the channel access time should not exceed,
\begin{equation}
\bar{\alpha}^{k} \leq \frac{-\ln (1- \overline{p_{coll}^k}/P_{\mbox{\footnotesize \textit{Inactive}}}^k)}{\mu_k}.
\end{equation}
Furthermore, $\bar{\alpha}^{k} $ is bounded by the duration of the data transmission phase $T - \tau_s$. Thus,
\begin{equation}
\bar{\alpha}^{k} = \min \left(\frac{-\ln(1- \overline{p_{coll}^k}/P_{\mbox{\footnotesize \textit{Inactive}}}^k)}{\mu_k}, T - \tau_s \right).
\end{equation}

Let $\bf{T}$ and $\bf{P}$ with elements $t_{n,k}$ and $p_{n,k}$ denote the transmission time and power allocation matrices of size $N \times \bar{K}$. Let $t_{n,k}$ and $p_{n,k}$ denote the transmission time and power of the $n$-th data sensor over the $k$-th channel, respectively. The total energy consumption of the data sensors is determined by
\begin{equation}\label{jatp_a}
 \sum_{n=1}^N \sum_{k=1}^{\bar{K}} t_{n,k} p_{n,k}.
\end{equation}
The transmission time of all of the data sensors over the $k$-th channel is limited by the channel access time $\bar{\alpha}^k$,
\begin{equation} \label{jatp_b}
\sum_{n=1}^N t_{n,k} \leq \bar{\alpha}^k, \forall k.
\end{equation}
Furthermore, the transmission time of the $n$-th data sensor is bounded by the duration of the data transmission phase, namely,
\begin{equation}\label{jatp_c}
\sum_{k=1}^{\bar{K}} t_{n, k} \leq T-\tau_s, \forall n.
\end{equation}

The data amount that is required from the $n$-th data sensor is denoted by $D_{n}$. During the data transmission phase, the $n$-th data sensor transmits sensed data over the $k$-th channel to the sink at a transmission power of $p_{n,k}$ and duration of $t_{n, k}$. The data transmission rate is given by
\begin{equation}
\label{eqn_shannon}
R_{n,k} = W \log_2 \left(1+ \delta_{n,k} p_{n,k} \right),
\end{equation}
where $\delta_{n,k}$ represents the $n$-th sensor channel gain over the $k$-th channel at the sink. The allocated rate should be sufficiently large to support the generated data. This relationship is captured by
\begin{equation}\label{jatp_d}
\sum_{k=1}^{\bar{K}} t_{n, k} R_{n,k} \geq D_n.
\end{equation}
The transmission time $t_{n,k}$ and power $p_{n,k}$ are nonnegative. Additionally, $p_{n,k}$ is constrained by the maximum transmission power $p_{\max}$. Thus, we have
\begin{align}
&t_{n,k}  \geq 0, \forall k, ~\forall n ~\mbox{and} \label{jatp_e} \\
&0  \leq p_{n,k} \leq p_{\max} \label{jatp_f}.
\end{align}

We allocate the transmission time $\bf{T}$ and power $\bf{P}$ to minimize the energy consumption of all of the data senors, which can be formulated as:
\begin{align*}
&(\mbox{DSRA})~ \min_{ \bf{T, P}} \sum_{n=1}^N \sum_{k=1}^{\bar{K}} t_{n,k} p_{n,k} \\ \nonumber
{\rm s.t.}~ &\begin{cases}
  \sum_{n=1}^N t_{n,k} \leq \bar{\alpha}^k, \forall k,\\
  \sum_{k=1}^{\bar{K}} t_{n, k} \leq T-\tau_s, \forall n,\\
  \sum_{k=1}^{\bar{K}} t_{n, k} W \log_2(1 + \delta_{n,k} p_{n,k}) \geq D_n, \forall n, \\
  t_{n,k} \geq 0, \forall k, n, \\
  0 \leq p_{n,k} \leq p_{\max}, \forall k, n.
\end{cases}
\end{align*}

The amount of data to transmit is determined by the product of the transmission time $t_{n,k}$ and logarithm of the power $p_{n,k}$. These structures lead to the non-convexity of the problem DSRA with potentially multiple local optima and generally implies difficulty in determining the global optimal solution \cite{Floudas1990}. However, by showing that DSRA is biconvex, we gain access to algorithms that efficiently solve biconvex problems \cite{Gorski2007}; see Appendix \ref{appendix-biconvex}.
\subsubsection{Joint Time and Power Allocation (JTPA) Algorithm}
Because DSRA is biconvex, the variable space is divided into two disjoint subspaces.  Therefore, the problem is divided into two convex subproblems that can be solved efficiently: time allocation (DSRA-1) and power allocation (DSRA-2). The time allocation problem is given by
\begin{subequations} \notag
\label{opt_givenP}
\begin{align}
 \mbox{(DSRA-1)}~&\min_{\bf{T}}
   \begin{aligned}[t]
        \sum_{n=1}^N \sum_{k=1}^{\bar{K}} t_{n,k} p_{n,k}
   \end{aligned}  \\
   \text{s.t.~~} &
     (\ref{jatp_b}) (\ref{jatp_c}) (\ref{jatp_d}) (\ref{jatp_e}), \notag
\end{align}
\end{subequations}
while the power allocation problem is given by,
\begin{subequations} \notag
\label{opt_givenT}
\begin{align}
\mbox{(DSRA-2)}~ &\min_{\bf{P}}
   \begin{aligned}[t]
        \sum_{n=1}^N \sum_{k=1}^{\bar{K}} t_{n,k} p_{n,k}
   \end{aligned}  \\
   \text{s.t.~~} &
     (\ref{jatp_d}) (\ref{jatp_f}), \notag
\end{align}
\end{subequations}
In the following, we adopt the Alternate Convex Search in \cite{Gorski2007} to solve the DSRA problem. In every step of the proposed algorithm, one of the variables is fixed, and the other is optimized, and vice versa in the subsequent step. The proposed algorithm solves the two problems iteratively and converges to a partially optimal solution\footnote{The definition of a partial optimal solution is given as follows:
\begin{defn}
\label{def_paropt}
Let $f : S \to \mathbb{R}$ be a given function and let $(x^*, y^*) \in S$. Therefore, $(x^*, y^*)$ is called a partial optimum of $f$ on $S$, if
\begin{eqnarray*}
 f(x^*, y^*) & \leq & f(x, y^*) \quad \forall x \in S_{y^*} \nonumber ,  \\
 f(x^*, y^*) & \leq & f(x^*, y) \quad \forall y \in S_{x^*}.
\end{eqnarray*}
\end{defn} $S_{y^*}$ and $S_{x^*}$ denote the $y^*$- and $x^*$-sections of $S$ \cite{Gorski2007}.}. The detailed procedure of the proposed algorithm is given as follows:
\begin{algorithm}[h]
\caption{Proposed Algorithm JTPA}
\label{alg_acs}
\begin{algorithmic}[1]
\REQUIRE Network parameters, stopping criterion $\epsilon$ and maximum number of  iterations $i_{\max}$.
\ENSURE The optimal (${\bf T}^*, {\bf P}^*$).
\STATE Choose an arbitrary starting point (${\bf T}^0, {\bf P}^0$) and set the iteration index as $i = 0$, and the initial solution as $z^0 = 0$;
\REPEAT \STATE {Fix ${\bf P}^i$ and determine the optimal ${\bf T}^{i+1}$ by solving DSRA-1 via the Simplex method \cite{Dantzig2003}; \label{al2_s1}}
\STATE {Fix ${\bf T}^{i+1}$, determine the optimal ${\bf P}^{i+1}$ and objective function value $z_{i}$ by solving DSRA-2 via the Interior Point method \cite{Boyd2004}; \label{al2_s2}}
\STATE {$i = i + 1$};
\UNTIL{$z^{i+1} - z^{i-1} < \epsilon$ or $i \geq i_{\max}$}
\RETURN $({\bf T}^{i+1}, {\bf P}^{i+1})$
\end{algorithmic}
\end{algorithm}

The convergence of the proposed algorithm to the global optimum is not guaranteed since DSRA is biconvex and could have several local optima. However, because the objective function is differentiable and biconvex over a biconvex set, convergence to a stationary point that is partially optimal is guaranteed \cite{Gorski2007}. Data sensors transmit their data to the sink using the transmission time and power that is determined by the proposed JTPA algorithm.
\section{Performance Evaluation} \label{sec:evaluation}
We evaluate the performance of the C-E algorithm in the spectrum sensing phase and the JTPA algorithm in the data transmission phase through performing simulations. The simulation results are obtained through Matlab on a computer with intel core(TM) i7-4510u CPU@2.00GHz 2.6GHz, 8 GB RAM.
\subsection{Simulation Setup} \label{sec:par_setting}
We simulate an HCRSN that consists of $M = 10$ spectrum sensors and $N = 30$ data sensors. The sensors are randomly placed in a circular area with a radius of $20$ meters. The sink is located at the center of this circular area. The HCRSN coexists with a primary network that is deployed over an area that has a radius of $200$ meters. The PUs' transmission power is 1 mW, and the noise power is $-80$ dB. The PU's channel gain at the sensor is simulated based on $1/d^{3.5}$, where $d$ is the distance between the PU and spectrum sensor. The target false alarm probability for all of the spectrum sensors $\bar{p}_f$ is set to $0.1$. PUs transmit QPSK modulated signals, with each over a $6$ MHz bandwidth $W$. The default number of licensed channels is seven unless specified otherwise. Over the seven channels, seven PUs operate over one channel exclusively. Their transition rates $\lambda_k, k=1, \cdots, 7$, are $0.6, 0.8, 1, 1.2, 1.4, 1.6, 1.8$, respectively. Additionally, the transition rates $\mu_k, k=1, \cdots, 7$, are $0.4, 0.8, 0.6, 1.6, 1.2, 1.4, 1.8$, respectively. The network operates periodically over slots of length $T=100$ ms \cite{Pei2011} (see Fig. \ref{fig_tml}). The maximum transmission power is set to $p_{max} = 100$ mW \cite{Shu2006}. The remaining parameters are set according to Table \ref{table.parametersettings} unless specified otherwise. In the following two subsections, we evaluate the performances of the proposed algorithms.

\begin{table}[!t]
    \caption{Parameter Settings}
    \centering
    \small
    \begin{tabular}{p{6.5cm}|p{1.3cm}}
         \hline
         \hline 
         \textbf{Parameter} & \textbf{Settings} \\
         \hline
         \hline
         Upper bound of the transmission power $p_{max}$ & 100 mW\\
         Bandwidth of the licensed channel $W$ \cite{Zhang2014} & 6 MHz\\
         Sampling rate of the EH spectrum sensor $U$ \cite{Liang2008}  & 6000\\
         False alarm probability $\bar{p}_f$ & 0.1  \\
         Energy consumption per spectrum sensing  & $0.11$ mJ\\ 
         Time consumption per spectrum sensing $\tau_{s^{'}}$   & 1 ms   \\
         Duration of the spectrum-sensing phase $\tau_s$   & 5 ms   \\
         Duration of each slot $T$ \cite{Pei2011} & 100 ms \\
         Upper bound of the collision probability $p_{coll}^k(\bar{\alpha}^{k})$ & 0.1 \\
         Mis-detection probability threshold $\overline{MD}_{thr}$ & 0.9 \\
         Fraction of samples retained in C-E $\rho$ & 0.6 \\
         Stopping threshold of C-E $\epsilon$ & $10^{-3}$ \\
         \hline
    \end{tabular}
    \label{table.parametersettings}
\end{table}

\subsection{Performance Evaluation of the C-E Algorithm}
The following simulation results provide insights into the performance of the C-E algorithm over the spectrum sensing phase. Metrics of interest include the convergence speed and quality of the obtained solution. Furthermore, we study the impact of the stopping criterion parameters $\rho$ and $\epsilon$ on those metrics. The performance of the proposed algorithm is also compared to the performance of a candidate greedy algorithm.

In Fig. \ref{SSS_O_R}, we show the optimality of the C-E algorithm in a scenario that has 3 spectrum sensors and 2-4 licensed channels. We reduce the number of spectrum sensors in such a way that an exhaustive search can be efficiently performed. The EH rate and sensing time are set to be sufficiently large that any assignment would be feasible. The C-E algorithm's optimal solution, i.e., the Detected Average Available Time of Channels (DAATC), is compared to that obtained by random assignment and exhaustive search. The random assignment randomly assigns licensed channels to the spectrum sensors, while the exhaustive search traverses all of the possible assignments. As shown in Fig. \ref{SSS_O_R}, the expected detected channel's available time obtained by the C-E algorithm is close to that of the exhaustive search and is able to achieve $87\% - 94\%$ of it. The proposed algorithm's computed solution is 2 to 3 times larger than that of the random assignment.


\begin{figure}[h]
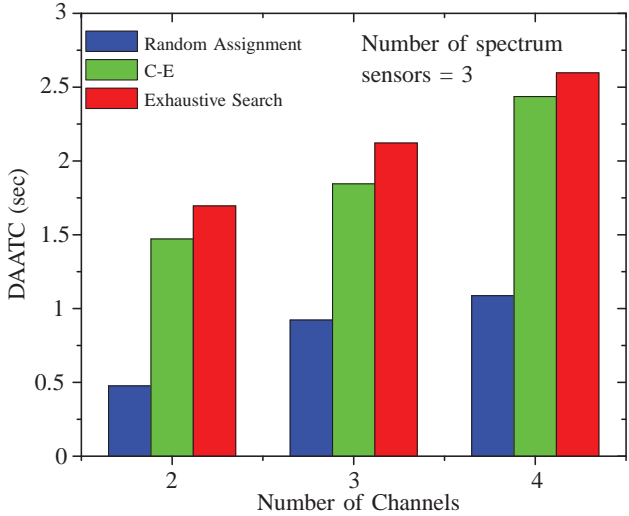
 
\centering
\begin{lpic}[l(8mm),r(5mm),t(0mm),b(5mm)]{Figure_4(0.38,0.38)}
\small
\lbl[l]{26,148; \scriptsize Random Assignment}
\lbl[l]{26,139; \scriptsize C-E}
\lbl[l]{26,129;\scriptsize  Exhaustive Search}
\lbl[l]{102,148;\small Number of spectrum}
\lbl[l]{102,138;\small sensors = 3}
\lbl[.]{101,-12; Number of Channels}
\lbl[.]{-18,82,90; DAATC (sec)}

\lbl[b]{37,-7; 2 }
\lbl[b]{101,-7;3 }
\lbl[b]{165,-7;4 }

\lbl[r]{2,4; 0 }
\lbl[r]{1,30; 0.5 }
\lbl[r]{1,56; 1 }
\lbl[r]{0,82; 1.5}
\lbl[r]{0,108; 2}
\lbl[r]{0,134; 2.5}
\lbl[r]{-1,160; 3}
\end{lpic}
\caption{The comparison of C-E algorithm's performance and the performance of random assignment and exhaustive search in terms of the DAATC.}
\label{SSS_O_R}
\end{figure}

For a network of 10 spectrum sensors with 7 channels, the stability of the C-E algorithm is shown in Figs. \ref{Fig:CE_energy_cvg} and  \ref{Fig:CE_dur_cvg}. Fig. \ref{Fig:CE_energy_cvg} shows that the convergence of the C-E algorithm with respect to the EH rate ranges from $3$ mW to $7$ mW\footnote{In \cite{Zhang2013}, the real experimental data obtained from the Baseline Measurement System (BMS) of the Solar Radiation Research Laboratory (SRRL) shows that, the EH rate ranges from 0 mW to 100 mW for most of the day.}. It can be seen that the value of the objective function fluctuates during the startup phase and then converges to the maximum DAATC after 30 iterations. Moreover, the value of the objective function increases by one-third for the case in which EH rate = $7$ mW, while it doubles for the case in which the EH rate = $3$ mW. This finding demonstrates the responsiveness of the stochastic policy updating strategy defined by Eqn. (\ref{C-D_Update}). Moreover, it can be clearly seen that the DAATC increases with the EH rate.

Fig. \ref{Fig:CE_dur_cvg} shows the convergence results for the C-E algorithm with respect to the spectrum-sensing duration $\tau_s$ range of 2 ms to 6 ms and EH rate of 7 mW. As we can see from the figure, the value of DAATC fluctuates at the startup phase. This is because the samples of channel assignment vectors are generated according to the uniform distribution at the initialization step of the C-E algorithm. As the C-E algorithm executes, the probability to generate samples that bring higher DAATC increases. At last, the algorithm converges to a stable solution that leads to highest DAATC in 30 iterations. Furthermore, the DAATC increases with the length of the spectrum sensing phase $\tau_s$, because more channels can be detected by the spectrum sensors with larger $\tau_s$.

\begin{figure}[!tbp]
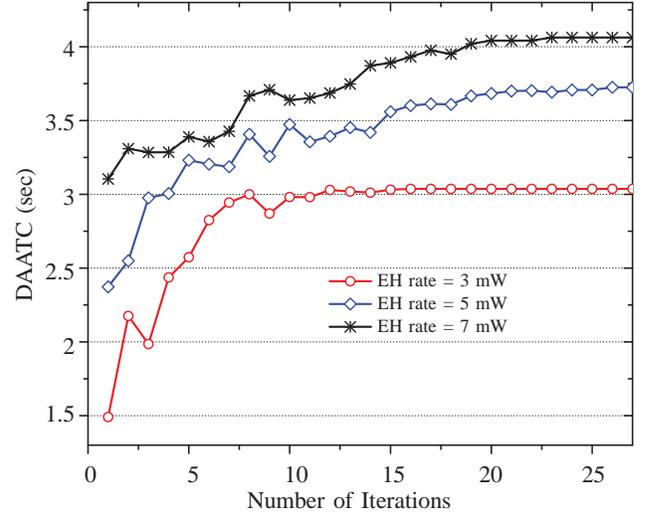

  \centering
    \begin{lpic}[l(9mm),r(5mm),t(5mm),b(5mm)]{Figure_5(0.38,0.38)}
    \small
    \lbl[l]{105,62;\scriptsize EH rate = 3 mW}
    \lbl[l]{105,54;\scriptsize EH rate = 5 mW}
    \lbl[l]{105,46;\scriptsize EH rate = 7 mW}
    \lbl[.]{95,-15;   Number of Iterations}
    \lbl[.]{-18,78,90;  DAATC (sec)}

    \lbl[b]{6,-9;  0 }
    \lbl[b]{42,-9;  5 }
    \lbl[b]{78,-9;  10 }
    \lbl[b]{112,-9;  15 }
    \lbl[b]{146,-9;  20 }
    \lbl[b]{182,-9;  25 }

    \lbl[r]{2,15;  1.5 }
    \lbl[r]{2,39;  2 }
    \lbl[r]{2,65;  2.5 }
    \lbl[r]{2,91;  3 }
    \lbl[r]{2,117;  3.5 }
    \lbl[r]{2,143;  4 }
    \end{lpic}
    \caption{Convergence of the C-E algorithm for three different EH rates, $\pi_m$}
\label{Fig:CE_energy_cvg}
  \end{figure}

  \begin{figure}[h]
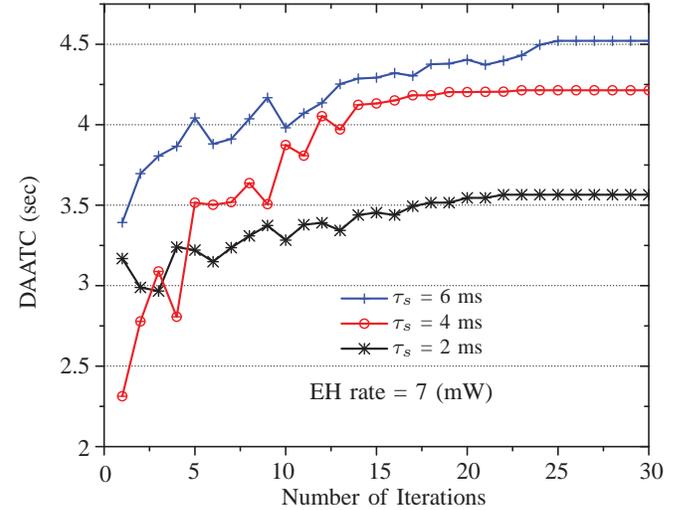

    \begin{lpic}[l(11mm),r(5mm),t(5mm),b(5mm)]{Figure_6(0.38,0.38)}
    \small
    \lbl[l]{105,56; \footnotesize  $\tau_s$ = 6 ms}
    \lbl[l]{105,47; \footnotesize  $\tau_s$ = 4 ms}
    \lbl[l]{105,38; \footnotesize  $\tau_s$ = 2 ms}
    \lbl[.]{102,-14;  Number of Iterations}
    \lbl[l]{76,22;   EH rate = 7 (mW)}
    \lbl[.]{-22,75,90; \small DAATC (sec)}

    \lbl[b]{6,-9;  0 }
    \lbl[b]{37,-8;  5 }
    \lbl[b]{68,-8;  10 }
    \lbl[b]{102,-8;  15 }
    \lbl[b]{133,-8;  20 }
    \lbl[b]{164,-8;  25 }
    \lbl[b]{197,-8;  30 }

    \lbl[r]{2, 3;  2 }
    \lbl[r]{2,32;  2.5 }
    \lbl[r]{2,61;  3 }
    \lbl[r]{2,88;  3.5 }
    \lbl[r]{2,117;  4 }
    \lbl[r]{2,145;  4.5 }
    \end{lpic}
    \caption{Convergence of the C-E algorithm for three different spectrum-sensing durations, $\tau_s$.}
    \label{Fig:CE_dur_cvg}
\end{figure}

The C-E algorithm stops iterating if the inequality in Eqn. (\ref{eqn_threshold}) holds, or the maximum number of iterations is reached. Figs. \ref{Fig:SSS_eps} and \ref{Fig:SSS_rho} show the impact of fine tuning the algorithm parameters, $\epsilon$ and $\rho$, on the convergence speed and quality of the obtained solution. It can be seen from Fig. \ref{Fig:SSS_eps} that a large number of iterations is required to satisfy the stopping criterion, and a larger DAATC can be obtained for a small $\epsilon$. Furthermore, the algorithm converges in less than 100 iteration even for the small value of $\epsilon = 10^{-6}$. Fig. \ref{Fig:SSS_rho} shows the impact of the fraction of samples that is retained (i.e., $\rho$) in each step on the algorithm performance. The C-E algorithm converges faster with small $\rho$. Moreover, the DAATC peaks at one value of $\rho$ and then starts falling. For the parameters that considered in this study, $\rho$ peaks at $0.6$. The fraction $\rho$ should be optimized to obtain a larger DAATC.

\begin{figure}[h]
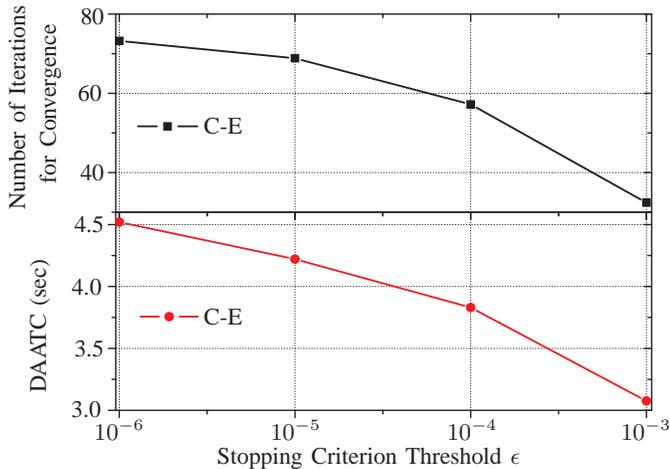
 
  \centering
   \begin{lpic}[l(12mm),r(5mm),t(5mm),b(5mm)]{Figure_7(0.34,0.33)}
   \small
    \lbl[l]{39,119;  C-E}
    \lbl[l]{39,42; C-E}
    \lbl[.]{105,-15;   Stopping Criterion Threshold $\epsilon$ }
    \lbl[.]{-34,128,90;   Number of Iterations }
    \lbl[.]{-21,128,90;    for Convergence}

    \lbl[.]{-25,40,90;   DAATC (sec) } 

    \lbl[b]{6,-8;  $10^{-6}$}
    \lbl[b]{75,-8;  $10^{-5}$}
    \lbl[b]{144,-8;  $10^{-4}$}
    \lbl[b]{212,-8;  $10^{-3}$}

    \lbl[r]{2,4.5;  3.0 }
    \lbl[r]{0,29;  3.5}
    \lbl[r]{-1,54;  4.0}
    \lbl[r]{0,79;  4.5}

    \lbl[r]{0,100;  40}
    \lbl[r]{-1,132;  60}
    \lbl[r]{0,163;  80}


\end{lpic}
\caption{The effect of $\epsilon$ on the performance of the C-E algorithm.}
\label{Fig:SSS_eps}
\end{figure}

  \begin{figure}[h]
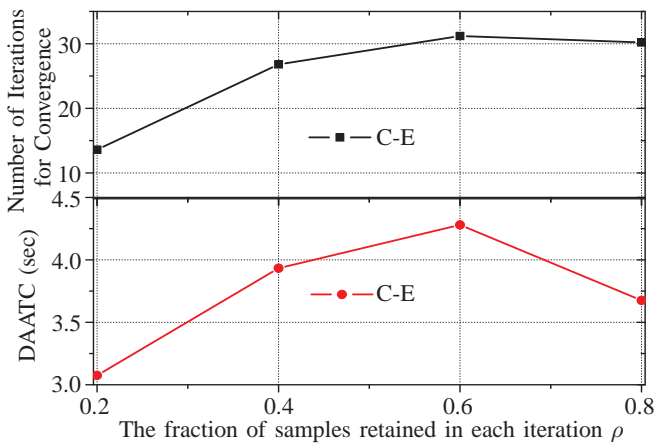

  \centering
  \begin{lpic}[l(9mm),r(5mm),t(5mm),b(5mm)]{Figure_8(0.35,0.33)}
  \small
    \lbl[l]{112,41;  C-E} 
    \lbl[l]{112,104;  C-E}
    \lbl[.]{110,-15;    The fraction of samples retained in each iteration $\rho$}
    \lbl[.]{-25,116,90;   Number of Iterations } 
    \lbl[.]{-14,116,90;    for Convergence} 

    \lbl[.]{-20,40,90;   DAATC (sec)}

    \lbl[b]{5,-8;  0.2}
    \lbl[b]{75,-8;  0.4}
    \lbl[b]{143,-8;  0.6}
    \lbl[b]{212,-8;  0.8}
    \lbl[r]{3,4;  3.0 }
    \lbl[r]{0,29;  3.5}
    \lbl[r]{0,54;  4.0}
    \lbl[r]{0,79;  4.5}

    \lbl[r]{3,89; 10 }
    \lbl[r]{0,115; 20}
    \lbl[r]{0,141; 30}

\end{lpic}
\caption{The impact of the fraction of retained samples $\rho$ on the performance of the proposed C-E algorithm.}
\label{Fig:SSS_rho}
  \end{figure}

  \begin{figure}[!tbp]
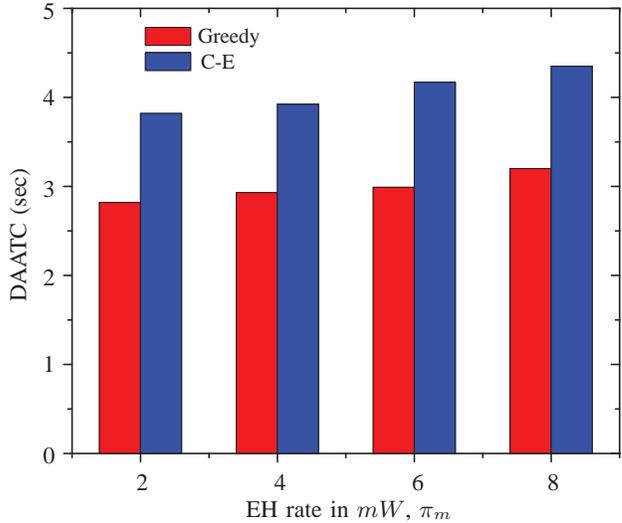

  \centering
  \begin{lpic}[l(9.5mm),r(5mm),t(5mm),b(5mm)]{Figure_9(0.382,0.382)}
    \small
    \lbl[l]{48,148.5; \footnotesize Greedy }
    \lbl[l]{49,140.5;\footnotesize C-E }
    \lbl[.]{102,-16; \small  EH rate in $mW$, $\pi_m$ }
    \lbl[.]{-14,80,90; \small DAATC (sec)}

    \lbl[b]{30,-9;2 }
    \lbl[b]{78,-9;4 }
    \lbl[b]{126,-9;6 }
    \lbl[b]{173,-9;8 }

    \lbl[r]{1,3;0 }
    \lbl[r]{1,35;1 }
    \lbl[r]{1,65;2 }
    \lbl[r]{1,96;3 }
    \lbl[r]{1,128;4 }
    \lbl[r]{1,158;5 }
    \end{lpic}
	\caption{A comparison of the C-E algorithm and the Greedy algorithm performance for a range of EH rates.}
    \label{Fig:CE_gr_EHrate}
\end{figure}

\begin{figure}[h]
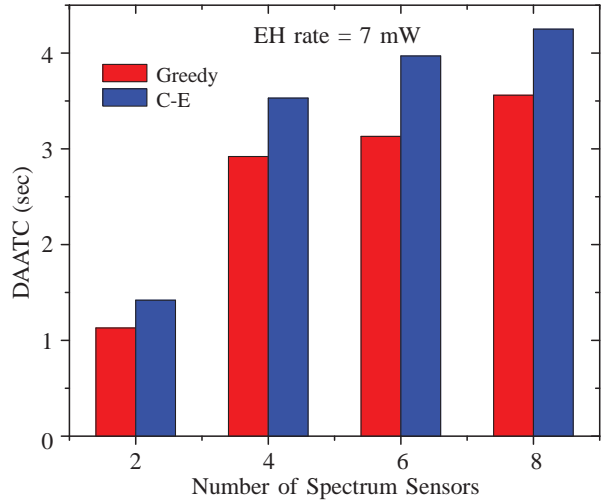
 
\begin{lpic}[l(10mm),r(5mm),t(5mm),b(5mm)]{Figure_10(0.37,0.37)}
    \small
    \lbl[l]{35,133;\footnotesize Greedy }
    \lbl[l]{35,124.7;\footnotesize C-E }
    \lbl[.]{101,-15;\small  Number of Spectrum Sensors }
    \lbl[.]{-13,77,90; \small DAATC (sec)}
    \lbl[l]{70,148;\small  EH rate = 7 mW }
    \lbl[b]{29,-8;2 }
    \lbl[b]{77,-8;4 }
    \lbl[b]{125,-8;6 }
    \lbl[b]{173,-8;8 }

    \lbl[r]{1,2;0 }
    \lbl[r]{1,38;1 }
    \lbl[r]{1,73;2 }
    \lbl[r]{1,107;3 }
    \lbl[r]{1,141;4 }
    \end{lpic}
	\caption{A comparison of the C-E algorithm and the Greedy algorithm performance for a number of spectrum sensors.}
    \label{Fig:CE_gr_numnode}
\end{figure}

Figs. \ref{Fig:CE_gr_EHrate} and \ref{Fig:CE_gr_numnode} show the comparison between the performance of the C-E algorithm and that of the greedy algorithm. The greedy algorithm corresponds to the algorithm proposed in \cite{Yu2011}; it picks the spectrum sensors sequentially and assigns them the channels that bring the largest DAATC. It can be seen from Fig. \ref{Fig:CE_gr_EHrate} that the C-E algorithm outperforms the greedy algorithm in terms of the obtained DAATC over a range of EH rates. A similar result can be seen in Fig. \ref{Fig:CE_gr_numnode}, where the number of spectrum sensors varies for a fixed EH rate of 7 mW.

\subsection{Performance Evaluation of the JTPA Algorithm}

\begin{figure}[h]
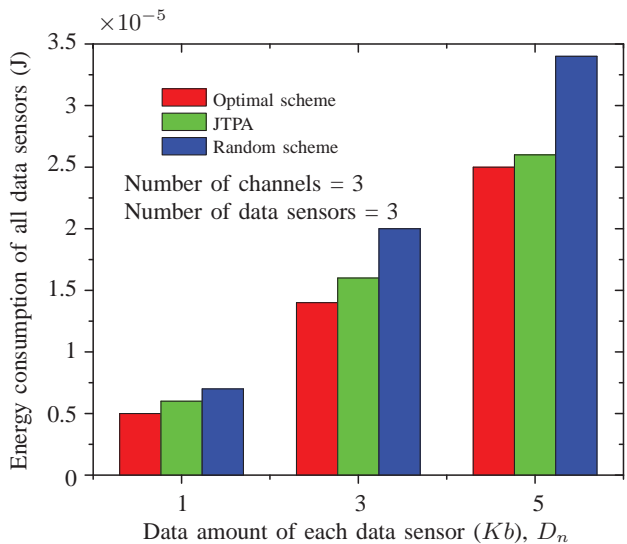

  \centering
    \begin{lpic}[l(9mm),r(5mm),t(5mm),b(5mm)]{Figure_11(0.37,0.37)}
    \small
    \lbl[l]{47,138; \scriptsize Optimal scheme}
    \lbl[l]{47,130; \scriptsize JTPA}
    \lbl[l]{47,122;\scriptsize Random scheme}
    \lbl[l]{15,109;  Number of channels = 3}
    \lbl[l]{15,99;  Number of data sensors = 3}
    \lbl[.]{100,-17;  Data amount of each data sensor ($Kb$), $D_n$ }
    \lbl[.]{-22,81,90;  Energy consumption of all data sensors (J)}
    \lbl[b]{36,-9; 1}
    \lbl[b]{100,-9; 3}
    \lbl[b]{164,-9; 5}
    \lbl[r]{1,3; 0 }
    \lbl[r]{1,26; 0.5 }
    \lbl[r]{-2,49; 1}
    \lbl[r]{-2,71; 1.5}
    \lbl[r]{-2,93; 2}
    \lbl[r]{-2,116; 2.5}
    \lbl[r]{-2,137; 3}
    \lbl[r]{-2,160; 3.5}
    \lbl[b]{19.5,163;$\times 10^{-5}$ }

    \end{lpic}
	\caption{A comparison of JTPA with the random scheme and optimal scheme}
    \label{Fig:JATP_O_R}
\end{figure}

\begin{figure}
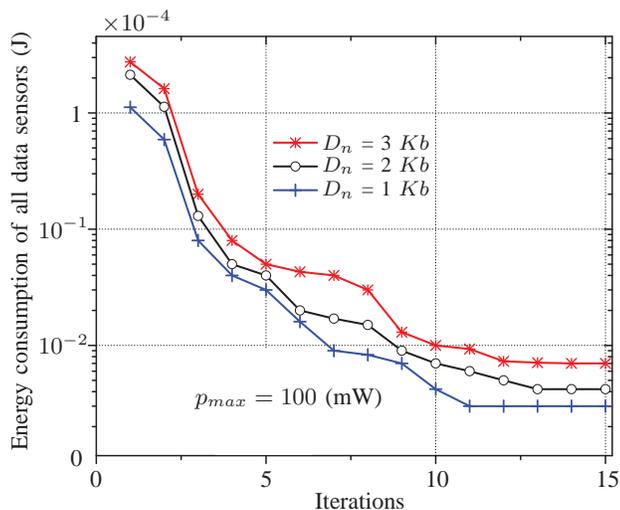

\centering
    \begin{lpic}[l(8mm),r(5mm),t(5mm),b(5mm)]{Figure_12(0.36,0.36)}
    \small
    \lbl[l]{87,119; \footnotesize $D_n$ = 3 $Kb$}
    \lbl[l]{87,111; \footnotesize $D_n$ = 2 $Kb$}
    \lbl[l]{87,102; \footnotesize $D_n$ = 1 $Kb$}
    \lbl[l]{40,25;$p_{max} = 100$ (mW) }
    \lbl[.]{101,-13;Iterations}
    \lbl[.]{-24,83,90;Energy consumption of all data sensors (J)}

    \lbl[b]{5,-8;0 }
    \lbl[b]{68,-8;5 }
    \lbl[b]{131,-8;10 }
    \lbl[b]{194,-8;15 }

    \lbl[r]{1.5,3;0 }
    \lbl[r]{4,45;$10^{-2}$ }
    \lbl[r]{4,88;$10^{-1}$ }
    \lbl[r]{1.5,130;1 }
    \lbl[b]{20,160;$\times 10^{-4}$ }
    \end{lpic}
    	\caption{The convergence of JTPA for a three data amounts, $D_n$ = 1,2,3 $Kb$. }
        \label{Fig:JATP_cvg}
\end{figure}

In this subsection, we evaluate the performance of the JTPA algorithm. For a network of three data sensors with three channels, we first verify the optimality of JTPA by comparing its performance to that of random scheme and optimal scheme in Fig. \ref{Fig:JATP_O_R}. The random scheme randomly assigns channels to the data sensors, while the optimal scheme searches over the complete space. Once the channels are assigned, a Matlab optimization toolbox is used to allocate the time and power. As shown in Fig. \ref{Fig:JATP_O_R}, JTPA consumes $5\%$ to $14\%$ more energy than the optimal scheme. However, JTPA conserves $18\%$ to $31\%$ more energy than the random assignment scheme.

The convergence of JTPA is evaluated in a network of ten spectrum sensors and thirty data sensors with five channels. The transition rates $\lambda_{1:5}$ and $\mu_{1:5}$ are set to $0.6, 0.8, 1, 1.2, 1.4$ and $0.4, 0.8, 0.6, 1.6, 1.2$, respectively. The spectrum sensing duration $\tau_s$ is set to $5$ ms. Fig. \ref{Fig:JATP_cvg} shows the convergence performance of JTPA with respect to the data amount ($D_n$) that is transmitted from each data sensor to the sink. It can be observed that the JTPA algorithm converges after 10 iterations and the energy consumption decreases $97\%$ during the first 6 iterations which implies the efficiency of the proposed algorithm.

\begin{figure}[h]
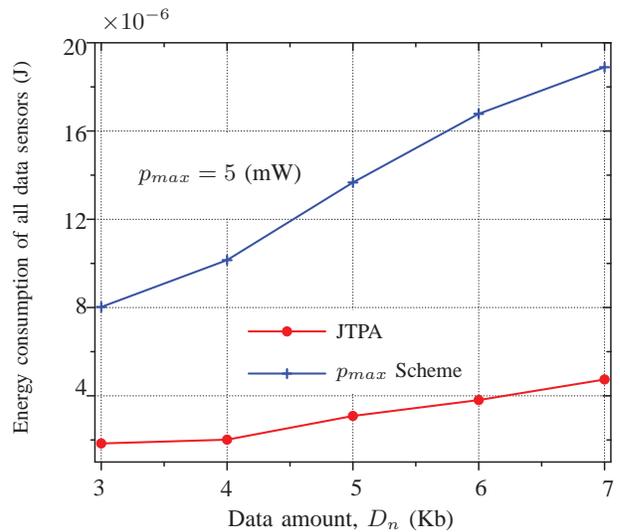

  \centering
   \begin{lpic}[l(8mm),r(5mm),t(5mm),b(5mm)]{Figure_13(0.36,0.36)}
   \small
    \lbl[l]{93,52; \footnotesize JTPA}
    \lbl[l]{93,37; \footnotesize $p_{max}$ Scheme}
    \lbl[l]{20,110;  $p_{max} = 5$ (mW) }
    \lbl[l]{53,-18;  Data amount, $D_n$ (Kb) }
    \lbl[.]{-23,83,90; \footnotesize Energy consumption of all data sensors (J)}

    \lbl[b]{7,-9;  3 }
    \lbl[b]{54,-9;  4 }
    \lbl[b]{102,-9;  5 }
    \lbl[b]{148,-9;  6 }
    \lbl[b]{195,-9;  7 }

    \lbl[r]{4,31.5;  4 }
    \lbl[r]{4,61.5;  8 }
    \lbl[r]{4,93;  12 }
    \lbl[r]{4,125;  16 }
    \lbl[r]{4,157;  20 }

    \lbl[r]{35,167; $\times 10^{-6}$ }
    \end{lpic}
    	\caption{A comparison of JTPA and the $p_{max}$ scheme for a range of data amounts.}
       \label{Fig:ACS_power}
  \end{figure}

  \begin{figure}[h]
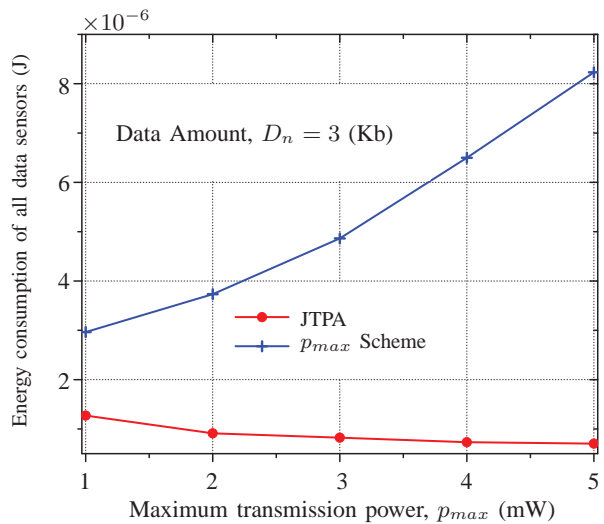

  \centering
    \begin{lpic}[l(10mm),r(5mm),t(5mm),b(5mm)]{Figure_14(0.36,0.36)} 
    \small
    \lbl[l]{85,54; \footnotesize  JTPA}
    \lbl[l]{85,45; \footnotesize $p_{max}$ Scheme}
    \lbl[l]{17,122;  Data Amount, $D_n=3$ (Kb) }
    \lbl[.]{102,-17; Maximum transmission power, $p_{max}$ (mW) }
    \lbl[.]{-18,83,90; \footnotesize  Energy consumption of all data sensors (J)}

    \lbl[b]{7,-9;  1 }
    \lbl[b]{54,-9;  2 }
    \lbl[b]{102,-9;  3 }
    \lbl[b]{148,-9;  4 }
    \lbl[b]{195,-9;  5 }

    \lbl[r]{3,31;  2 }
    \lbl[r]{3,68;  4 }
    \lbl[r]{3,104;  6 }
    \lbl[r]{3,140;  8 }

    \lbl[r]{33,164; $\times 10^{-6}$ }
    \end{lpic}
    	\caption{A comparison of JTPA and the $p_{max}$ Scheme for various $p_{max}$ values.}
        \label{Fig:ACS_dataamount}
\end{figure}


In Figs. \ref{Fig:ACS_power} and Fig. \ref{Fig:ACS_dataamount}, we compare the energy consumption of data transmission under the JTPA algorithm and the $p_{max}$ scheme. In the $p_{max}$ scheme, the data sensors transmit at the maximum available power $p_{max}$, and the transmission time is determined by solving the linear programming problem JTPA-1. The $p_{max}$ scheme is comparable to the channel allocation scheme proposed in \cite{Byun2008}, in which data sensors transmit data at fixed transmission power. Fig. \ref{Fig:ACS_power} shows the comparison of the energy consumption performance with respect to various required amount of data, while $p_{max}$ is set to 5 mW. Because the JTPA algorithm jointly allocates the transmission time and power over the available channels, JTPA consumes less energy than $p_{max}$ scheme for different data amount. Fig. \ref{Fig:ACS_dataamount} shows the comparison of the JTPA algorithm against the $p_{max}$ scheme for various values of $p_{max}$ and data amount $D_n = 3$ Kb $\forall n$. The energy consumption of the JTPA algorithm decreases with an increase in $p_{max}$ because data sensors can adjust the transmission power in a larger space for a larger $p_{max}$. Similar to the results shown in Fig. \ref{Fig:ACS_power}, JTPA consumes less energy than that of the $p_{max}$ scheme due to the joint allocation of transmission time and power.

\section{Conclusions} \label{sec:conlusion}

In this paper, a novel resource allocation solution for heterogeneous cognitive radio sensor networks (HCRSNs) has been proposed. The proposed solution assigns channels to spectrum sensors in such a way that the detected available time of the channels is maximized. Furthermore, it efficiently allocates the available channels to the data sensors along with the transmission time and power in order to prolong their lifetime. Extensive simulation results have demonstrated the optimality and efficiency of the proposed algorithms. The solution presented in this work enables using primary networks channels efficiently while adapting in real time to the availability of harvested energy, and optimizes the allocation of the battery-powered data sensors' scarce resources. This yields significantly higher spectral and energy-efficient HCRSNs.

For the future work, we plan to investigate the channel allocation and routing protocol design in EH-aided multi-hop HCRSNs, considering the time-varying EH rate and the adaptive detection threshold of sensors.

\section*{Acknowledgment}
This work was supported by the Fundamental Research Funds for the Central Universities of the Central South University (No. 2013zzts043). The project was supported partially by the Kuwait Foundation for the Advancement of Sciences under project code: P314-35EO-01. This work was also supported by National Natural Science Foundation of China (61379057, 61272149) and NSERC, Canada.

\appendices

\section{Derivation of the collision probability $p_{coll}^k(\bar{\alpha}^{k})$}
\label{Apn_col_prob}

Let $T_{\mbox{\footnotesize \textit{Inactive}}}^k$ be the sojourn time of a OFF/Inactive period with the probability density function (p.d.f) $f_{T_{\mbox{\footnotesize \textit{Inactive}}}^k}(\alpha)$. Given the exponentially distributed ON/OFF period, the p.d.f of the Inactive period is equal to \cite{Kim2008}
\begin{equation} \notag
f_{T_{\mbox{\footnotesize \textit{Inactive}}}^k}(\alpha) = \mu_k e^{-\mu_k \alpha}
\end{equation}
The probability that the OFF/Inactive period is less than $\bar{\alpha}^k$, i.e., the PU on channel $k$ returns in [0, $\bar{\alpha}^k$], can be derived to be
\begin{equation} \notag
\begin{split}
Pr(T_{\mbox{\footnotesize \textit{Inactive}}}^k < \bar{\alpha}^k) =& \int_0^{\bar{\alpha}^k} f_{T_{\mbox{\footnotesize \textit{Inactive}}}^k}(\alpha) \ud \alpha \\
=& 1 - e^{-\mu_k \bar{\alpha}^k}.
\end{split}
\end{equation}
Since channel $k$ is available with probability $P_{\mbox{\footnotesize \textit{Inactive}}}^k$, the probability of collision $p_{coll}^k(\bar{\alpha}^{k})$ is given by:
\begin{equation} \notag
p_{coll}^k(\bar{\alpha}^{k}) = P_{\mbox{\footnotesize \textit{Inactive}}}^k \cdot (1-e^{-\mu_k \bar{\alpha}^{k}}).
\end{equation}

\section{Bi-convexity of the DSRA problem}
\label{appendix-biconvex}
In the following we show that the DSRA problem is bi-convex.
\begin{theorem} \label{thm_biconvex}
If we fix one set of variables in $\bf{T}$ or $\bf{P}$, then DSRA is convex with respect to the other set of variables. Thus, DSRA is biconvex.
\end{theorem}
\begin{proof}
We first determine a feasible $\bf{P}$, and then, DSRA becomes a problem of determining $\bf{T}$ to satisfy
\begin{subequations} \notag
\label{opt_givenP}
\begin{align}
 \mbox{(DSRA-1)}~&\min_{\bf{T}}
   \begin{aligned}[t]
        \sum_{n=1}^N \sum_{k=1}^{\bar{K}} t_{n,k} p_{n,k}
   \end{aligned}  \\
   \text{s.t.~~} &
     (\ref{jatp_b}) (\ref{jatp_c}) (\ref{jatp_d}) (\ref{jatp_e}). \notag
\end{align}
\end{subequations}
which is linear and convex due to the linear objective function and linear feasible set. DSRA-1 can be solved using the Simplex method \cite{Dantzig2003}. Additionally, by fixing $\bf{T}$, DSRA becomes a problem of determining $\bf{P}$ to satisfy
\begin{subequations} \notag
\label{opt_givenT}
\begin{align}
\mbox{(DSRA-2)}~ &\min_{\bf{P}}
   \begin{aligned}[t]
        \sum_{n=1}^N \sum_{k=1}^{\bar{K}} t_{n,k} p_{n,k}
   \end{aligned}  \\
   \text{s.t.~~} &
     (\ref{jatp_d}) (\ref{jatp_f}), \notag
\end{align}
\end{subequations}
DSRA-2 can be solved by the interior point method. Both DSRA-1 and DSRA-2 are convex and can be solved efficiently. Therefore, the objective function $\sum_{n=1}^N \sum_{k=1}^{\bar{K}} t_{n,k} p_{n,k}$ is biconvex on the feasible set which makes DSRA a biconvex problem.
\end{proof}

\ifCLASSOPTIONcaptionsoff
  \newpage
\fi
\end{document}